\documentclass[pra]{revtex4}
\usepackage{amsmath}
\usepackage{amsfonts}
\usepackage{amssymb}
\usepackage{amsthm}
\usepackage{enumerate}
\usepackage{color}

\newcommand{\ket}[1]{|#1\rangle}

\newcommand{\proj}[1]{|#1\rangle\langle#1|}
\newcommand{\ba}{\mathbf{a}}
\newcommand{\bx}{\mathbf{x}}

\newcommand{\subs}{\mathcal{S}}

\newcommand{\tr}{\mathrm{tr}}

\newtheorem{theorem}{Theorem}
\newtheorem{lemma}{Lemma}

\newtheorem{corollary}{Corollary}

\begin{document} 
\title{Simulating all non-signalling correlations via classical or quantum theory with negative probabilities}

\author{Sabri W. Al-Safi} \email{S.W.Al-Safi@damtp.cam.ac.uk}
\affiliation{DAMTP, Centre for Mathematical Sciences, Wilberforce Road, Cambridge CB3 0WA, UK}

\author{Anthony J. Short} \email{tony.short@bristol.ac.uk}
\affiliation{H. H. Wills Physics Laboratory, University of Bristol, Tyndall Avenue, Bristol, BS8 1TL, UK}

\begin{abstract}  
Many-party correlations between measurement outcomes in general probabilistic theories are given by conditional probability distributions obeying the non-signalling condition. We show that any such distribution can be obtained from classical or quantum theory, by relaxing positivity constraints on either the mixed state shared by the parties, or the local functions which generate  measurement outcomes. Our results apply to generic non-signalling correlations, but in particular they yield two distinct quasi-classical models for quantum correlations.
\end{abstract} 

\maketitle

\section{Introduction}

Quantum theory predicts many strange phenomena, but non-local correlations are perhaps the most intriguing. On one hand, quantum theory is non-signalling: local measurements made by separate observers on a joint quantum state cannot convey information from one observer to another. However, the outcome statistics of these quantum measurements can correlate in such a way as to defy any classical explanation based on local influences \cite{bell65}.

Stronger non-locality often leads to greater aptitude in information-theoretic tasks \cite{vandam05, linden06, pawlowski09, almeida10}, and consequently much work has been done to understand the non-local power of quantum theory \cite{tsirelson80, brassard06, alsafi11, dahlsten12, pawlowski09.2}. Furthermore, the discovery of non-local correlations which are not achievable in quantum theory, yet still non-signalling \cite{pr94}, has provoked a more general study of non-local correlations and the  general probabilistic theories which generate them \cite{barrett07, short&barrett10, barrett05}.

In this paper, we will explore various models which are capable of generating any non-signalling correlation, yet which closely mirror the structure of classical and quantum theory. The key modification in each case is to drop a positivity requirement on one of the mathematical objects in the theory; this corresponds to allowing outcome probabilities to be negative for some unperformed measurements. Similar models have arisen elsewhere in the study of quantum mechanics - for example in the Wigner phase space representation of quantum states \cite{wigner32}, in analysis of the EPR paradox \cite{feynman82}, and in the context of quantum information theory \cite{ferrie11} - and there have also  been  attempts to provide a physical or mathematical interpretation to the notion of negative probabilities \cite{muckenheim86, khrennikov09}. Our analysis extends to correlations  arising from any non-signalling theory, and may likewise be of benefit in the study of these theories. 

In the context of local models based on classical probability theory, we will show that any non-signalling correlation can be generated if either the joint probability distribution over states, or the local conditional probability distributions associated with measurements, are allowed to be negative. In the former case, this is essentially an alternative and constructive proof of the  result that non-signalling states of a joint system lie in the affine hull of the pure product states \cite{barrett07, degorre11}. In the latter case, the result is more surprising, because only the positivity of the \emph{local} measurements is modified, yet arbitrary \emph{non-local} correlations can be generated.

In the context of quantum theory, Ac\'{i}n \emph{et al} \cite{acin10} have shown that the standard Born trace rule can be extended to generate any non-signalling correlation,  by allowing the quantum state (usually represented by a positive density operator) to be a non-positive operator. We employ our classical framework to give an alternative proof, and a slight extension, of this result. We also derive a dual result that any non-signalling correlation can be generated if the local quantum measurement operators  are allowed to be non-positive. 

 Ac\'{i}n \emph{et al}'s  result has recently been used in proving that any physical theory for which the local structure is identical to that of qubits, and which admits at least one continuous, reversible interaction, must have the global structure specified by quantum theory \cite{masanes12}. Likewise, the ability to represent the correlations of a broad class of theories in a similar way to quantum or classical correlations may provide a powerful tool in analysing such theories, and identifying the unique properties of quantum theory. 

\section{Setup} 

In order to study experimental correlations, an abstract framework is commonly used in which little is assumed about the underlying physics. Consider an experiment involving $N$ systems, in which an observer at system $k$ chooses one of a finite set of possible measurements, indexed by $x_k \in \{1,2,\ldots, X_k\}$, and records one of a finite set of possible outcomes, indexed by $a_k \in \{1,2,\ldots, A_k\}$. We will denote the set of possible values for $k$ by $[N] \equiv\{1,2,\ldots, N\}$. The experiment may be characterized by the conditional probability distribution $p(a_1,\ldots,a_N | x_1, \ldots, x_N)$ on the measurement outcomes, given the inputs. 

 The correlations obtained from such an experiment are \emph{local} if they can be generated by a local classical model, where each system has its own state which individually determines the probability of measurement outcomes on it, and where these states are distributed  according to some joint probability distribution.  Specifically, $p(a_1, \ldots, a_N| x_1, \ldots, x_N)$ is local if there exists a joint probability distribution $p_{\Lambda} (\lambda_1, \ldots ,\lambda_N)$ for the local state $\lambda_k \in \Lambda_k$ of each system $k$ (where each $\Lambda_k$ is a finite set), and a conditional probability distribution $p_k(a_k|x_k, \lambda_k)$ for each $k$, such that 

\begin{equation} \label{localrep}
p(a_1, \ldots, a_N| x_1, \ldots, x_N)= \sum_{\lambda_1, \ldots , \lambda_N}  p_1(a_1|x_1, \lambda_1) \cdots p_N(a_N|x_N, \lambda_N)  p_{\Lambda} (\lambda_1, \ldots, \lambda_N)  
\end{equation}

This definition of locality is equivalent to the more standard requirement that $p(a_1, \ldots, a_N| x_1, \ldots, x_N)$ lies in the convex hull of the product conditional probability distributions $p_1(a_1|x_1)\cdots p_N(a_N|x_N)$ \cite{barrett05}, or that there is a shared state $\lambda$, obtained with probability $p(\lambda)$, which deterministically specifies the outcomes of all local measurements \cite{bell65}. 

The discomfiting truth is that not all experiments on quantum states generate local correlations \cite{bell65}. Instead, the results of quantum experiments are given in general by the formula
\begin{equation}  \label{quantumrep}
p(a_1,\ldots,a_N | x_1, \ldots, x_N) = \tr \left( \left( M^{(1)}_{a_1|x_1} \otimes \cdots \otimes M^{(N)}_{a_N|x_N} \right) \rho \right), 
\end{equation}
where $\rho$ is a density matrix (satisfying $\rho\geq 0$ and $\tr(\rho)=1$), and $M^{(k)}_{a_k|x_k}$ are measurement operators comprising some positive-operator valued measure (POVM) for each $k$ (satisfying $M^{(k)}_{a_k|x_k} \geq 0$ and $\sum_{a_k} M^{(k)}_{a_k|x_k} = I^{(k)}$, the identity on system $k$). 

In section \ref{classresults} we show that all non-signalling correlations can be represented in the form of (\ref{localrep}), as long as either (i) the classical state $p_{\Lambda}$, or (ii) the classical measurement operators $p_k$ are allowed to contain negative components. We then show in section \ref{quantresults} that all non-signalling correlations can be represented in the form of (\ref{quantumrep}), as long as either (i) the quantum state $\rho$, or (ii) the quantum measurement operators $M^{(k)}_{a_k|x_k}$, are allowed to be non-positive operators.

\section{Classical results} \label{classresults}

In order to analyse the classical case, it will be helpful to define an analogue of a probability distribution, in which the entries are allowed to be negative; we will refer to this as a quasiprobability distribution. A function $\tilde{p}: C_1 \times \cdots \times C_k \rightarrow \mathbb{R}$, where $|C_i| < \infty$ for all $i$, is a \emph{joint quasiprobability distribution} if and only if it obeys the normalisation condition

\begin{equation} \label{normalization}
 \sum_{c_1, \ldots , c_N} \tilde{p}(c_1, \ldots, c_N) = 1 
\end{equation}

 (We will only consider joint distributions, and will therefore drop the use of the word `joint'). Similarly, a conditional quasiprobability distribution is a real function $\tilde{p}(c_1, \ldots, c_K | z_1, \ldots, z_L)$ which is a quasiprobability distribution for each fixed choice of $z_1, \ldots, z_L$. Note that any (conditional) quasiprobability distribution which is non-negative for all values of its arguments is also a (conditional) probability  distribution. For clarity we will use tildes throughout to represent quasiprobability distributions.

We wil say that a conditional quasiprobability distribution $\tilde{p}(a_1, \ldots, a_N| x_1, \ldots, x_N)$ is \emph{non-signalling} if, for any $k \in [N]$ and choice of $x_1, \ldots, x_N$, the sum $\sum_{a_k} \tilde{p}(a_1, \ldots, a_N| x_1, \ldots, x_N)$ is independent of the value of $x_k$. Non-signalling distributions have well-defined reduced distributions: for any non-empty subset $ \subs = \{i_1,\ldots ,i_M\} \subseteq [N]$ let $\ba_{\subs}$ and $\bx_{\subs}$ denote the reduced strings $(a_{i_1},\ldots ,a_{i_M})$, $(x_{i_1},\ldots, x_{i_M})$; the \emph{marginal distribution}  may then be defined as

\begin{equation}
\tilde{p}(\ba_{\subs}|\bx_{\subs}) = \sum_{a_i \,:\, i\notin\subs} \tilde{p}(a_1, \ldots, a_N| x_1, \ldots, x_N) 
\end{equation}

\noindent for some arbitrary choice of $x_i$ for $i \notin \subs$.
 
\begin{lemma}  \label{marginal-lemma}
 A non-signalling, conditional quasiprobability distribution $\tilde{p}(a_1, \ldots, a_N| x_1, \ldots, x_N)$ is uniquely characterized by the complete set of marginal distributions $\tilde{p}(\ba_{\subs}|\bx_{\subs})$ for which $a_{i_k} < A_{i_k} \, \forall i_k \in \subs$, where $\subs$ ranges over all subsets of $[N]$.
\end{lemma}
 
\begin{proof}
This lemma is an explicit statement in terms of quasiprobabilities of a result used in \cite{acin10}. Suppose that $\tilde{p}(a_1, \ldots, a_N| x_1, \ldots, x_N)$ and $\tilde{p}'(a_1, \ldots, a_N| x_1, \ldots, x_N)$ are non-signalling conditional quasiprobability distributions, whose marginals $\tilde{p}(\ba_{\subs}|\bx_{\subs})$ and $\tilde{p}'(\ba_{\subs}|\bx_{\subs})$ agree whenever $a_{i_k} < A_{i_k} \, \forall i_k \in \subs$, for all $\subs$. We argue that $\tilde{p}(\ba_{\subs}|\bx_{\subs})$ and $\tilde{p}'(\ba_{\subs}|\bx_{\subs})$ agree for all choices of $a_i$, by induction on $n = \#\{i_k \in \subs \, : \, a_{i_k} = A_{i_k}\}$. The case $n=0$ holds trivially; for $n>0$, without loss of generality,  suppose that $a_{i_1} = A_{i_1}$, then,
\begin{align}
 \tilde{p}(\ba_{\subs}|\bx_{\subs}) &= \tilde{p}(A_{i_1},\ldots,a_{i_M} | x_{i_1},\ldots,x_{i_M})\\
 &= \tilde{p}(a_{i_2},\ldots,a_{i_M} | x_{i_2},\ldots,x_{i_M}) - \sum_{a_{i_1} < A_{i_1}} \tilde{p}(a_{i_1},\ldots,a_{i_M} | x_{i_1},\ldots,x_{i_M}) \label{defmarg}\\
 &= \tilde{p}'(a_{i_2},\ldots,a_{i_M} | x_{i_2},\ldots,x_{i_M}) - \sum_{a_{i_1} < A_{i_1}} \tilde{p}'(a_{i_1},\ldots,a_{i_M} | x_{i_1},\ldots,x_{i_M}) \label{indhyp}\\
 &= \tilde{p}'(\ba_{\subs}|\bx_{\subs}) .
\end{align}
 Line \eqref{defmarg} follows from the definition of the marginal distribution, and line \eqref{indhyp} follows from applying the induction hypothesis to the $n-1$ case. Setting $\subs = [N] $ then proves the lemma. 
\end{proof}
 
 This lemma tells us that as long as the non-signalling property is obeyed, we can restrict our attention to a subset of all possible measurement outcomes: this is very helpful in proving the following Theorems.

\begin{theorem}(Non-positive classical measurements) An N-partite conditional probability distribution $p(a_1, \ldots, a_N| x_1, \ldots, x_N)$ is non-signalling if and only if it can be represented in the form \label{thm:measurements}
\begin{equation} \label{negativemeasurements} 
p(a_1, \ldots, a_N| x_1, \ldots, x_N)= \sum_{\lambda_1, \ldots, \lambda_N}  \tilde{p}_1(a_1|x_1, \lambda_1) \cdots \tilde{p}_N(a_N|x_N, \lambda_N)  p_{\Lambda} (\lambda_1, \ldots, \lambda_N)  
\end{equation}
where $ p_{\Lambda} (\lambda_1, \ldots, \lambda_N)$ is a probability distribution, and $\tilde{p}_k(a_k|x_k, \lambda_k)$ is a conditional quasiprobability distribution for each $k$. 
\end{theorem}

\begin{proof} Firstly, note that as the $\tilde{p}_k(a_k|x_k, \lambda_k)$ are conditional quasiprobability distributions, they must satisfy $\sum_{a_k} \tilde{p}_k(a_k|x_k, \lambda_k) =1$ (independent of $x_k$). It follows that the distribution given by (\ref{negativemeasurements}) is non-signalling. Conversely, to prove that all non-signalling distributions $\tilde{p}$ can be written in the form of \eqref{negativemeasurements}, we take each $\Lambda_k$ to be the set of ordered pairs $[a'_k, x'_k]$ consisting of the allowed measurement choices and outputs for system $k$. We then set
\begin{eqnarray}
p_{\Lambda}([a'_1, x'_1], \ldots,[a'_N, x'_N]) &=& \frac{p(a'_1, \ldots, a'_N | x'_1, \ldots, x'_N)}{X_1 X_2 \cdots X_N} \label{conds_1} \label{eq:proof1_1} \label{rdef1} \\
\tilde{p}_k (a_k|x_k, \lambda_k)  &=& \left\{ \begin{array}{ll} X_k \delta_{\lambda_k, [a_k,x_k]} &\qquad \textrm{if} \qquad a_k<A_k  \vspace{0.2cm}\\
1 -  \sum_{a<A_k} X_k  \delta_{\lambda_k ,[a,x_k]} &\qquad \textrm{if} \qquad \label{conds_2}
a_k=A_k \end{array} \right. \label{eq:proof1_2}
\end{eqnarray}

With these assignments, $p_{\Lambda}(\lambda_1, \ldots ,\lambda_N)$ is a probability distribution and $\tilde{p}_k(a_k|x_k, \lambda_k)$ is a conditional quasiprobability distribution. However, $\tilde{p}_k (A_k |x_k, [1,x_k]) = 1-X_k < 0$ whenever $X_k > 1$, hence $\tilde{p}_k(a_k|x_k, \lambda_k)$ will not usually be a valid conditional probability distribution. \smallskip

It remains to show that the given values for $p_{\Lambda}(\lambda_1, \ldots ,\lambda_N)$ and $\tilde{p}_k(a_k|x_k, \lambda_k)$ satisfy \eqref{negativemeasurements}. Consider the quasiprobability distribution $p'$ given by 
\begin{equation} \label{p'def}
p'(a_1, \ldots, a_N| x_1, \ldots, x_N) = \sum_{\lambda_1, \ldots, \lambda_N}  \tilde{p}_1(a_1|x_1, \lambda_1) \cdots \tilde{p}_N(a_N|x_N, \lambda_N)  p_{\Lambda}(\lambda_1, \ldots, \lambda_N) .
\end{equation}
It is straightforward to check that $p'$ is non-signalling and that $p(\ba_{\subs}|\bx_{\subs}) = p'(\ba_{\subs}|\bx_{\subs})$ for all subsets $\subs \subseteq [N]$, and strings $\ba_{\subs}$ and $\bx_{\subs}$ with $a_i < A_i \,\,\forall i \in \subs$. For example, when $\subs = \{1,2\}$ we have
\begin{eqnarray} 
p'(a_1, a_2| x_1, x_2) &=& \sum_{\lambda_1, \ldots, \lambda_N}X_1  \delta_{\lambda_1, [a_1,x_1]} X_2 \delta_{\lambda_2, [a_2,x_2]}  p_{\Lambda}(\lambda_1, \ldots, \lambda_N). \nonumber \\
&=& X_1 X_2 \frac{p(a_1, a_2| x_1, x_2)}{X_1 X_2} = p(a_1, a_2| x_1, x_2) .
\end{eqnarray} 

\noindent By lemma 1, this is enough to conclude that $p'$ is indeed the probability distribution $p$, hence \eqref{negativemeasurements} holds.
\end{proof}

\begin{theorem} (Non-positive classical states) An N-partite conditional probability distribution $p(a_1, \ldots, a_N| x_1, \ldots, x_N)$ is non-signalling if and only if it can be represented in the form \label{thm:states}
\begin{equation} \label{negativestates} 
p(a_1, \ldots, a_N| x_1, \ldots, x_N)= \sum_{\lambda_1, \ldots, \lambda_N}  p_1(a_1|x_1, \lambda_1) \cdots p_N(a_N|x_N, \lambda_N)  \tilde{p}_{\Lambda} (\lambda_1, \ldots, \lambda_N)  
\end{equation}
where $ \tilde{p}_{\Lambda} (\lambda_1, \ldots, \lambda_N)$ is a quasiprobability distribution, and  $p_k(a_k|x_k, \lambda_k)$ is a conditional probability distribution for each $k$. 
\end{theorem}

\begin{proof} 
As each $p_k(a_k|x_k, \lambda_k)$ is a conditional probability distribution, it is clear that summing over $a_k$ removes any dependence on $x_k$ on the right-hand side, hence the distribution $p(a_1, \ldots, a_N| x_1, \ldots, x_N)$ is non-signalling. It remains to be shown that we can represent any non-signalling distribution in the form of (\ref{negativestates}). To achieve this, we take each $\Lambda_k$ to be the set of ordered pairs $[a'_k, x'_k]$ as before, along with an additional value which we will refer to as $\xi_k$, so that there are $A_k \cdot X_k +1$ possible choices of $\lambda_k$ for each $k$. For a string $(\lambda_1,\ldots ,\lambda_N) \in \Lambda_1 \times \cdots \times \Lambda_N$, let $\subs = \{i\in [N] \,:\, \lambda_i \neq \xi_i \}$ (i.e.  the set of indices $i$ for which $\lambda_i \neq \xi_i$) and define
\begin{equation} \label{rdef2}
 \tilde{p}_{\Lambda}(\lambda_1,\ldots ,\lambda_N) = \left[\Pi_{i \notin\subs} (1-X_i) \right] p(\ba_{\subs} | \bx_{\subs})
\end{equation}
so that, for example,
\begin{eqnarray}
\tilde{p}_{\Lambda}([a'_1, x'_1], [a'_2, x'_2], \ldots, [a'_N, x'_N]) &=& p(a'_1, \ldots, a'_N | x'_1, \ldots, x'_N)  \\ 
\tilde{p}_{\Lambda}(\xi_1, [a'_2, x'_2], \ldots, [a'_N, x'_N]) &=& (1-X_1) p(a'_2, \ldots, a'_N | x'_2, \ldots, x'_N)   \\
 \vdots &=& \vdots  \\
\tilde{p}_{\Lambda}(\xi_1, \xi_2, \ldots, \xi_N) &=& (1-X_1) (1-X_2) \cdots (1-X_N) . 
 \end{eqnarray}

To show that this is a quasiprobability distribution, note that

\begin{align}
 \sum_{\lambda_1 , \ldots , \lambda_N} \tilde{p}_{\Lambda}(\lambda_1,\ldots ,\lambda_N) &= \sum_{\subs \subseteq [N]} \sum_{\ba_{\subs},  \bx_{\subs}} \left[\Pi_{i \notin\subs} (1-X_i) \right] p(\ba_{\subs} | \bx_{\subs}) \\
&= \sum_{\subs \subseteq [N]} \left[ \Pi_{i \notin\subs} (1-X_i)\right] \left[ \Pi_{j \in \subs} X_j \right] \label{binexp} \\
&= \Pi_{i \in [N]} \left(X_i + (1-X_i)\right) = 1 .
\end{align}

The $p_k (a_k|x_k, \lambda_k)$ are defined to be the following conditional probability distributions:
\begin{equation} 
p_k (a_k|x_k, \lambda_k) = \left\{ \begin{array}{ll}  \delta_{\lambda_k, [a_k,x_k]} &\qquad \textrm{if} \qquad a_k<A_k  \vspace{0.2cm}\\
1 -  \sum_{a<A_k}  \delta_{\lambda_k, [a,x_k]} &\qquad \textrm{if} \qquad \label{conds_2x}
a_k=A_k \end{array} \right. .
\end{equation}

As in the previous case, it is a straightforward check that substituting (\ref{rdef2}) and (\ref{conds_2x}) into (\ref{p'def}) gives a non-signalling conditional quasiprobability distribution whose marginals agree with all the marginals of the distribution $p(a_1, \ldots, a_N| x_1, \ldots, x_N)$ in the case $a_i < A_i \,\,\forall i$. Hence, again by lemma 1 the initial probability distribution is recovered. Note that the measurement probabilites $p_k (a_k|x_k, \lambda_k)$ always equal 0 or 1 in this case, hence the measurements are in fact deterministic. 
\end{proof}
 
\section{Quantum results} \label{quantresults}

By a careful construction of states and measurements, our classical results immediately imply quantum corollaries. In particular, we can use Theorem \ref{thm:measurements} to prove 

\begin{corollary} (Non-positive quantum measurements) An N-partite conditional probability distribution $p(a_1, \ldots, a_N| x_1, \ldots, x_N)$ is non-signalling if and only if it can be represented in the form 
\begin{equation} \label{eq:cor1}
p(a_1, \ldots, a_N| x_1, \ldots, x_N)=  \tr \left( \left( \tilde{M}^{(1)}_{a_1|x_1} \otimes \cdots \otimes \tilde{M}^{(N)}_{a_N|x_N} \right) \rho \right), 
\end{equation}
where $\tilde{M}^{(k)}_{a_k|x_k}$ are Hermitian operators satisfying $\sum_{a_k} \tilde{M}^{(k)}_{a_k|x_k} = I$ for each $k$, and $\rho$ is a density operator (satisfying $\rho \geq 0, \tr(\rho)=1$).  Furthermore, this representation can be chosen such that the operators  $\tilde{M}^{(1)}_{a_1|x_1} \otimes \cdots \otimes \tilde{M}^{(N)}_{a_N|x_N}$ and $\rho$ all commute. 
\end{corollary} 

\begin{proof} 
Summing the right-hand side of (\ref{eq:cor1}) over $a_k$, it is clear that any $p$ which can be represented in this way is non-signalling. To prove the converse, we use the results of Theorem \ref{thm:measurements}. To each system $k$, we assign a Hilbert space spanned by the orthonormal basis $\{ \ket{\lambda_k} | \lambda_k \in \Lambda_k  \}$, where $\Lambda_k$ is the set of all ordered pairs $[a_k, x_k]$. We then take 
\begin{eqnarray} 
\rho &=& \sum_{\lambda_1, \ldots, \lambda_N} p_{\Lambda}(\lambda_1, \ldots ,\lambda_N)\; \proj{\lambda_1} \otimes \proj{\lambda_2}\otimes \cdots \otimes \proj{\lambda_N}, \\
\tilde{M}^{(k)}_{a_k|x_k} &=& \sum_{\lambda_k} \tilde{p}_k (a_k|x_k, \lambda_k)\, \proj{\lambda_k}, 
\end{eqnarray} 
where $ p_{\Lambda}(\lambda_1, \ldots ,\lambda_N)$ and $\tilde{p}_k (a_k|x_k, \lambda_k)$ are given by (\ref{eq:proof1_1}) and (\ref{eq:proof1_2}) respectively. Substituting this choice of operators into \eqref{eq:cor1} leads to exactly the same equation as \eqref{negativemeasurements}: therefore the result follows directly from Theorem \ref{thm:measurements}.
\end{proof} 

Similarly, from Theorem \ref{thm:states} we may obtain a re-derivation of Ac\'{i}n \emph{et al}'s result \cite{acin10}, with the slight strengthening that all the operators may be chosen to commute.

\begin{corollary} (Non-positive quantum states) An N-partite conditional probability distribution $p(a_1, \ldots, a_N| x_1, \ldots, x_N)$ is non-signalling if and only if it can be represented in the form 
\begin{equation} \label{eq:cor2}
p(a_1, \ldots, a_N| x_1, \ldots, x_N)=  \tr \left( \left( M^{(1)}_{a_1|x_1} \otimes \cdots \otimes M^{(N)}_{a_N|x_N} \right) \tilde{\rho} \right), 
\end{equation}
where $M^{(k)}_{a_k|x_k}$ are POVM elements (positive operators satisfying $\sum_{a_k} M^{(k)}_{a_k|x_k} = I$) for each $k$, and $\tilde{\rho}$ is a Hemitian operator satisfying $\tr(\rho)=1$.  Furthermore this representation can be chosen such that the operators  $M^{(1)}_{a_1|x_1} \otimes \cdots \otimes M^{(N)}_{a_N|x_N}$ and $\tilde{\rho}$ all commute. 
\end{corollary} 

\begin{proof} 
Again we see that any such $p$ is non-signalling by summing the right-hand side of \eqref{eq:cor2} over $a_k$. To prove the converse, we assign a Hilbert space spanned by the orhonormal basis $\{ \ket{\lambda_k} | \lambda_k \in \Lambda_k  \}$ to  each system $k$, where $\Lambda_k$ is the set of all ordered pairs $[a_k, x_k]$ and the extra element $\xi_k$  (as in the proof of Theorem 2). We then take 
\begin{eqnarray} 
\tilde{\rho} &=& \sum_{\lambda_1, \ldots, \lambda_N} \tilde{p}_{\Lambda}(\lambda_1, \ldots ,\lambda_N)\; \proj{\lambda_1} \otimes \proj{\lambda_2}\otimes \cdots \otimes \proj{\lambda_N}, \\
M^{(k)}_{a_k|x_k} &=& \sum_{\lambda_k} p_k (a_k|x_k, \lambda_k)\, \proj{\lambda_k}, 
\end{eqnarray} 
where $ \tilde{p}_{\Lambda}(\lambda_1, \ldots ,\lambda_N)$ and $p_k (a_k|x_k, \lambda_k)$ are given by (\ref{rdef2}) and (\ref{conds_2x}) respectively. Substituting these into (\ref{eq:cor2}) and applying Theorem \ref{thm:states} proves the corollary.
\end{proof} 

In both these theorems, the set of quantum correlations can be recovered by adding the requirement that the operators $\tilde{M}^{(k)}_{a_k|x_k}$ or $\tilde{\rho}$ are positive respectively. 

\section{Discussion}

 One way of viewing the mixture of states $p_{\Lambda}$ of a local distribution \eqref{localrep} is as a result of one's ignorance of the particular value of $(\lambda_1, \ldots, \lambda_N)$. From this perspective, {the results of Theorem \ref{thm:measurements} are particularly surprising: in this case the state $p_{\Lambda}(\lambda_1, \ldots, \lambda_N)$ is a standard mixture of product states (which one would normally think of as `local'), and all the measurement distributions $\tilde{p}_k (a_k |x_k, \lambda_k)$ are local objects. One might therefore wonder how this can generate non-local correlations at all. The explanation is that the measurements do not yield positive outcome probabilities for each component of the mixed state. Hence we can no longer think of the state as an ignorance mixture of allowed local states, but rather as a non-local object itself. 

As it stands, there is no discernible difference in using our procedures to generate a local correlation, compared with using them to generate a non-local correlation. In all cases the resulting distributions or operators  will contain negative components, as long as at least one party has more than one measurement choice. It would be interesting to see whether there exists a procedure similar to ours which outputs genuine probability distributions when generating local correlations. The question also arises whether quantum correlations might have some special status when represented in the form of classical quasiprobability representations as in (\ref{negativemeasurements}) and (\ref{negativestates}): perhaps, for example, the negative values in the distributions can be bounded if the correlation is quantum. However, this seems likely to be difficult, given the apparent difficulty of finding simple conditions for the quantum realizability of correlations\cite{navascues08}.

It is interesting to examine the efficiency of our representations. Note that the number of hidden variables we use is $\Pi_{k=1}^N A_k X_k$ in Theorem  \ref{thm:measurements} and  $\Pi_{k=1}^N  (A_k X_k + 1)$ in  Theorem \ref{thm:states}.  This can be reduced by noting that whenever $\lambda_k$ is equal to either $\xi_k$ or $[A_k, x_k]$, each local conditional probability distribution takes the same form: $p(a_k|x_k, \lambda_k) = \delta_{a_k, A_k}$. Hence this specific set of states can be combined into a single state $\eta_k$. To preserve normalization, whenever $\lambda_k = \eta_k$ the probabilities \eqref{rdef1} and \eqref{rdef2} must be summed over all the combined states. We do not use this compression of the state space in the main presentation of the theorems for clarity. However, this would reduce the total number of states in both classical theorems to  $\Pi_{k=1}^N  ( (A_k-1) X_k + 1)$, which is the same as the number of real parameters used to specify a non-signalling probability distribution via its marginals in Lemma  \ref{marginal-lemma} (in practice, one less parameter is needed due to normalisation).  In the quantum case, our representations use a quadratically larger Hilbert space than that of \cite{acin10}, with the payoff that all the operators commute.

We have shown that by taking either classical probability theory or quantum theory as a starting point, and by relaxing a positivity constraint on \emph{either}  the statistics of measurement outcomes, \emph{or}  the probabilistic mixing of states, one generates all non-signalling correlations. It would be interesting to investigate whether other general probabilistic theories \cite{barrett05} can be modified in a similar manner to yield all non-signalling correlations, or whether this is particular to theories which, like quantum theory, contain classical theory as a special case.

\bigskip 
 \textbf{Acknowledgments.} AJS acknowledges support from the Royal Society. SWA is funded by an EPSRC grant, and thanks Paul Skrzypczyk and James Yearsley for helpful discussions.

\end{document}